\documentclass{eptcs}

\usepackage{breakurl}
\usepackage{amsmath}
\usepackage{amssymb}
\usepackage{amsthm}
\usepackage{eucal}

\usepackage{ifdraft}                % Introduces \ifdraft

\newcounter{fact}
\newtheorem{theorem}[fact]{Theorem}
\newtheorem{lemma}[fact]{Lemma}
\newtheorem{definition}[fact]{Definition}
\newtheorem{example}[fact]{Example}

\usepackage[center]{subfigure}

\newcommand{\Ev}{E}
\newcommand{\Princ}{\mathcal{A}}
\newcommand{\aname}{\mathcal{A}}
\newcommand{\cname}{\mathcal{C}}

\newcommand{\princsym}{\pi}
\newcommand{\princ}[2][]{\princsym_{#1}({#2})}

\newcommand{\oksym}{\it ok}
\newcommand{\ok}[2]{{#1} \;\oksym\, {#2}}
\newcommand{\duties}[2]{\mathit{duties}({#1},{#2})}

\newcommand{\conf}[1][]{\mathcal{F}_{{#1}}}

\newcommand{\powset}[1]{\wp(#1)}

\newcommand{\mkset}[1]{\overline{#1}}

\renewcommand{\epsilon}{\varepsilon}
\newcommand{\hidden}[1]{}

\newcommand{\coco}{\mbox{\ensuremath{\mathrm{CO}_2}\hspace{2pt}}}

\newcommand{\pmv}[1]{\ensuremath{\mathsf{#1}}}

\newcommand{\atom}[1]{\textit{#1}}

\newcommand{\coimp}{\twoheadrightarrow}

\newcommand{\irule}[2]{
  \begin{array}{c}
    #1  \\ \hline
    #2
  \end{array}}

\newcommand{\imp}{\rightarrow}

\newcommand{\ask}[2]{\mathsf{ask}_{{#1}}\,{#2}}
\newcommand{\fuse}[2]{\mathsf{fuse}_{{#1}}\,{#2}}
 \newcommand{\says}{\ensuremath{\;\mathit{says}\;}}

\newcommand{\pcl}{\textup{PCL\;}}

\newcommand{\pclminus}{\ensuremath{1\mathit{N}\text{-}\pcl}}

\newcommand{\setenum}[1]{\{#1\}}
\newcommand{\setcomp}[2]{\{{#1} \;\mid\; {#2}\}}

\newcommand{\mytitle}[0]{An event-based model for contracts}

\title{\mytitle}

\author{Massimo Bartoletti \qquad\qquad
Tiziana Cimoli
\qquad\qquad
G.~Michele Pinna
\institute{Universit\`a degli Studi di Cagliari, Italy 
\email{\setenum{bart,t.cimoli,gmpinna}@unica.it} 
}
\and
Roberto Zunino
\institute{Universit\`a degli Studi di Trento and COSBI, Italy
\email{roberto.zunino@unitn.it}
}
}

\def\titlerunning{\mytitle}
\def\authorrunning{M. Bartoletti, T. Cimoli, G.M. Pinna, R. Zunino}

\begin{document}

\maketitle
\begin{abstract}
We introduce a basic model for contracts. 
Our model extends event structures 
with a new relation, 
which faithfully captures the circular dependencies among 
contract clauses. 
We establish whether an agreement exists which respects all the
contracts at hand
(i.e.\ all the dependencies can be resolved), and we detect
the obligations of each participant.
The main technical contribution is a correspondence between our model 
and a fragment of the contract logic PCL~\cite{BZ10lics}. 
More precisely, we show that the reachable events are exactly those which
correspond to provable atoms in the logic.
Despite of this strong correspondence, our model improves~\cite{BZ10lics}
by exhibiting a finer-grained notion of culpability, which 
takes into account the legitimate orderings of events.
\end{abstract}

\newcommand{\cem}[2]{#2}

\section{Introduction}

Contracts will play an increasingly important role in the
specification and implementation of distributed systems. 
Since participants in distributed systems may be mutually distrusted, and
may have conflicting individual goals, the possibility that
a participant behaviour may diverge from the expected one is quite
realistic. 
To protect themselves against possible misconducts,
 participants should postpone actual collaboration until reaching an
\emph{agreement} on the mutually offered behaviour. 
This requires a preliminary step, where each participant
declares her promised behaviour, i.e.\ her \emph{contract}.

A contract is a sort of
assume/guarantee rule, which makes explicit the dependency between
the actions performed by a participant, and those promised in return
by the others.
Event structures~\cite{Winskel86} can provide a basic semantic model
for assume/guarantee rules, by interpreting the enabling 
$\atom b \vdash \atom a$ as the contract clause:
``I will do {\atom a} \emph{after} you have done {\atom b}''.
However, event structures do not capture a typical aspect of contracts,
i.e.\ the capability of reaching an agreement when the assumptions
and the guarantees of the parties mutually match.
For instance, in the event structure with enablings
$\atom b \vdash \atom a$ and $\atom a \vdash \atom b$,
none of the events {\atom a} and {\atom b} is reachable,
because of the circularity of the constraints.
An agreement would still be possible if one of the parties 
is willing to accept a weaker contract. 
Of course, the contract \mbox{``I will do {\atom b}''} 
(modelled as $\vdash {\atom b}$)
will lead to an agreement with the contract ${\atom b} \vdash {\atom a}$, 
but it offers no \emph{protection} to the participant who offers it:
indeed, such contract can be stipulated without having anything in return.
% in the presence of an untrusted contract broker.

In this paper we introduce a model for contracts, by extending
(conflict-free) event structures with a new relation~$\Vdash$.
The contract $a \Vdash b$ (intuitively,
``I will do {\atom a} if you \emph{promise} to do {\atom b}'')
reaches an agreement with the dual contract $b \Vdash a$,
while protecting the participant who offers it.
We formalise agreements as configurations where all the participants
have reached their goals. 
We show that the problem of deciding if an agreement exists 
can be reduced to the problem of proving a suitable formula in 
(a fragment of) the contract logic \pcl\!~\cite{BZ10lics}, 
where an effective decision procedure for provability exists. 

Once an agreement has been found, the involved participants may safely 
cooperate by performing events.
Indeed, we prove that 
--- even in the presence of dishonest participants which do 
not respect their promises --- either all the participants reach their goals, 
or some of them is \emph{culpable} of not having performed her duties. 
A culpable  participant may then be identified (and possibly punished).
Also the problem of detecting duties and identifying culpable
participants is related to provability in \pcl\!.
Notably, while \pcl does not distinguish between the immediate duties 
and those that will only be required later on in a computation
(all provable atoms are considered duties in \pcl\!),
the richer semantical structure of our model allows for 
a finer-grained notion of duties, which depend on the actual
events already performed in a contract execution.
% computed  at each step of the contract execution.

\section{Contract model} \label{sect:contract-model}

A contract (Def.~\ref{def:contracts}) comprises a set of events $E$ and a set of participants $\aname$.
Each event $e \in E$ is uniquely associated to a participant $\princ{e} \in \aname$. 
Events are ranged over by $\atom{a}, \atom{b}, \ldots$, 
sets of events by $C,D,X,Y,\ldots$, and participants by ${\pmv A}, {\pmv B}, \ldots$.
Events are constrained by two relations: 
one is the enabling relation $\vdash$ of~\cite{Winskel86}, 
while the other is called \emph{circular} enabling relation, and it is denoted by $\Vdash$.
Intuitively, $D \vdash e$ states that $e$ may be performed \emph{after} 
all the events in $D$ have happened;
instead, $D \Vdash e$ means that $e$ may be performed either
if $D$ has already happened (similarly to~$\vdash$),
or possibly ``on credit'', 
on the promise that the events in $D$ will be performed at some later time.
The goals of each participant are indicated by the relation $\oksym$:
$\ok{\pmv A}{X}$ means that $\pmv A$ is satisfied if
\emph{all} the events in $X$ have happened.
The composition of contracts is defined component-wise, 
provided that events are uniquely associated to participants.

\begin{definition} \label{def:contracts}
A contract $\cname$ is a 6-tuple $(\Ev, \Princ, \princsym, \oksym, \vdash,\Vdash)$, 
where:
\begin{itemize}

\item $\Ev$ is a finite set of \emph{events};

\item $\Princ$ is a finite set of \emph{participants};

\item $\princsym: \Ev \rightarrow \Princ$ associates each event to a participant;

\item $\oksym \subseteq \Princ \times \powset{\Ev}$ is the 
\emph{fulfillment} relation, such that
\(
\ok{\pmv A}{X} \,\land\,  X \subseteq Y \implies \ok{\pmv A}{Y}
\);

\item $\vdash \;\subseteq \powset{\Ev} \; \times \; \Ev$ is the  
\emph{enabling relation};

\item $\Vdash \;\subseteq \powset{\Ev} \; \times \; \Ev$ is the 
\emph{circular enabling} relation.

\end{itemize}
We assume that both the enabling relations are \emph{saturated}, i.e.\
$X \circ e  \; \land \;  X \subseteq Y \implies Y \circ e$,
for $\circ \in \setenum{\vdash, \Vdash}$.
\end{definition}

The saturation of the relation $\oksym$ models the fact that
once a contract has been fulfilled
(i.e.\ a state is reached where all participants say $\oksym$), 
additional events can be neglected.

For notational convenience, 
we shall sometimes omit curly brackets around singletons, 
e.g.\ we shall write $a \vdash b$ instead of $\setenum{a} \vdash b$,
and we shall simply write $\vdash e$ for $\emptyset \vdash e$.
Similar abbreviations apply to $\Vdash$.

\begin{example} \label{ex:toys:contract}
Suppose there are three kids who want to play together.
Alice has a toy airplane, Bob has a bike, while Carl has a toy car.
Each of the kids is willing to share his toy, but they have different constraints:
Alice will lend her airplane only \emph{after} Bob has 
allowed her ride his bike;
Bob will lend his bike only after  he has played with Carl's car;
Carl will lend his toy car if the other two kids promise that
they will eventually let him play with their toys.
These constraints are modelled by the following contract~$\cname$,
where we only indicate the minimal elements of the relations $\vdash,\Vdash$ and $\oksym$:
% and the components $\Ev$, $\Princ$, and $\princsym$ are omitted for brevity:
\[
\begin{array}{lllll}
  & \Ev = \setenum{\atom{a},\atom{b},\atom{c}}
  \qquad
  & \setenum{\atom b} \vdash {\atom a}
  \qquad
  & \setenum{\atom c} \vdash {\atom b}
  \qquad
  & \setenum{\atom a, \atom b} \Vdash {\atom c} 
  \\[5pt]
  & \aname = \setenum{{\pmv A}, {\pmv B}, {\pmv C}}
  & \ok{\pmv A}{\setenum{\atom b}}
  & \ok{\pmv B}{\setenum{\atom c}}
  & \ok{\pmv C}{\setenum{\atom a,\atom b}}
  \\[5pt]
  % & \princsym = \setenum{\bind{\atom a}{\pmv A},\bind{\atom b}{\pmv B},\bind{\atom c}{\pmv C}} \\
  &
  & \princ{a} = {\pmv A}
  & \princ{b} = {\pmv B}
  & \princ{c} = {\pmv C}
\end{array}
\]
\end{example}

In the previous example, it is crucial that Carl's contract allows 
the event {\atom c} to happen ``on credit'' before the other events
are performed.
We shall show that this leads to an agreement among the participants,
while no agreement exists were Carl requiring 
$\setenum{\atom a, \atom b} \vdash {\atom c}$
(cf.\ Ex.~\ref{ex:toys:conf}).

In Def.~\ref{def:conf} we refine the notion of configuration 
of~\cite{Winskel86}, so to deal with the new $\Vdash$-enablings.
A set of events $C$ is a configuration if its events can be ordered
in such a way that each event $e \in C$ 
is either $\vdash$-enabled by its predecessors, 
or it is $\Vdash$-enabled by the whole $C$.
Configurations play a crucial role,
as they represent sets of events where all the debts have been honoured.

\begin{definition} \label{def:conf}
For all contracts $\cname$,
we say that $C \subseteq \Ev$ is a \emph{configuration of $\cname$} iff  
\[
  \exists e_0, \ldots, e_n. \;\;  
  \big(
  \setenum{e_0, \ldots, e_n} = C 
  \;\land\; 
  \forall i \leq n.\;\; 
  (
  \setenum{e_0, \ldots, e_{i-1}} \vdash e_i  
  \;\lor\;  
  C \Vdash e_i 
  )
  \big)
\]
The set of all configurations of $\cname$ is denoted by $\conf[\cname]$.
\end{definition}

\begin{example} \label{ex:handshaking:conf}
Not all sets of events are also configurations.
For instance, in the contract with enablings
$\atom a \Vdash \atom b$ and $\atom b \Vdash \atom a$, 
the sets $\emptyset$ and $\setenum{\atom a, \atom b}$ are configurations 
(in the latter, the use of $\Vdash$ allows for resolving 
the circular dependency between {\atom a} and {\atom b}),
while $\setenum{\atom a}$ and $\setenum{\atom b}$ are not.
\end{example}

\begin{example} \label{ex:toys:conf}
The contract $\cname$ of Ex.~\ref{ex:toys:contract}
has configurations $\emptyset$ and $\Ev = \setenum{\atom a,\atom b,\atom c}$.
Note that if Carl replaces his contract with
$\setenum{\atom a, \atom b} \vdash {\atom c}$,
then $\Ev$ no longer belongs to $\conf[\cname]$.
\end{example}

Following the examples above we observe that, 
differently from other event-based models,
if $C$ is a configuration, not necessarily $X \subseteq C$ is a configuration
as well. 
Hereafter, subsets of $\Ev$ are called \emph{states},
regardless they are configurations or not.

% In those example we have seen that an $X$-configuration 
% is also an $Y$-configuration when $X\subseteq Y$,  
% but we are interested in the minimal set 

Since our contracts have no conflicts (unlike~\cite{Winskel86}),
the union of two configurations is a configuration as well.

\newcommand{\lemconfplusconf}{
For all contracts $\cname$, 
if $C  \in \conf[\cname]$ and $D \in \conf[\cname]$,  
then $C \cup D \in  \conf[\cname]$.
}
\begin{lemma} \label{lem:conf-plus-conf}
\lemconfplusconf
\end{lemma}

Given a configuration $C$ and an event $e$, the set $C \cup \setenum{e}$
is still a configuration if $C \vdash e$ or $C \Vdash e$.
Otherwise, $C \cup \setenum{e}$ is not a configuration.
Compositional reasoning on sets of events (not necessarily configurations)
requires to keep track of the events taken ``on credit'',
as sketched in the proof of Th.~\ref{th:es-pcl}.

An event is \emph{reachable} when it belongs to a configuration; 
a set of events $X$ is reachable if every event in $X$ is reachable.
% When an event/set of events is $\emptyset$-reachable, we will just say
% that it is \emph{reachable}.
%
A reachable set is not necessarily a configuration 
(e.g.\ $\setenum{\atom{a},\atom{b}}$ in Ex.~\ref{ex:toys:contract});
yet, there always exists a configuration that contains it.
This follows by Lemma~\ref{lem:conf-plus-conf}, 
which guarantees that configurations are closed by union. 
The set comprising all the reachable events is a configuration
(actually, it is the greatest one).

\newcommand{\lemallinconf}{
Let $X \subseteq \Ev$ be a reachable set of events. Then,
\(
  \exists C \in \conf[\cname]. \;  X \subseteq C 
\).
}
\begin{lemma} \label{lem:all-in-conf} 
\lemallinconf
\end{lemma}

\begin{lemma} \label{lem:conf:reachable-maximal}
% For all $\cname$,  
$C = \bigcup \setcomp{e \in \Ev}{e \text{ is reachable}} \in \conf[\cname]$,
and $\forall C' \in \conf[\cname]. \;\; C' \subseteq C$.
%$\nexists C' \in \conf[\cname]$ such that $C' \supsetneq C$.
\end{lemma}

% \paragraph{Agreements.}
\subsection{Agreements}

Informally, a contract admits an \emph{agreement} when all the involved participants 
are happy with the guarantees provided by that contract.
In Def.~\ref{def:agreement}, we formalise an agreement on a contract $\cname$
as a configuration of $\cname$
where all the participants have reached their individual goals.
\text{E.g.}, the configuration $\Ev = \setenum{\atom a,\atom b,\atom c}$
is an agreement on the contract $\cname$ of Ex.~\ref{ex:toys:contract},
since $\ok{P}{\Ev}$ holds for $P \in \setenum{\pmv A,\pmv B,\pmv C}$ 
by saturation of $\oksym$.
% which fulfills all of their goals
% (cf.\ Ex.~\ref{ex:toys:conf}).

\begin{definition} \label{def:agreement}
An \emph{agreement} on $\cname$ is a configuration 
$C \in \conf[\cname]$ such that
$\forall \pmv A \in \Princ : \ok{\pmv A}{C}$.
\end{definition}

% A participant accepts a contract when some of her goals are included
% in a $\emptyset$-configurations.

% La definizione~\ref{def:agreement} qui sotto forse non serve.
% \begin{definition} \label{def:agreement}
% A participant {\pmv A} \emph{accepts} a contract $\cname$ iff,
% for some $X \subseteq \Ev$ such that $\ok{\pmv A}{X}$,
% there exists $C \in \conf[\cname]$ such that $X \subseteq C$.
% \end{definition}

We now establish the duties of a participant in a state where
some events $X$ have been performed.
Although several different definitions of duties are possible,
% e.g.\ establishing priorities among events, participants, \emph{etc}.
the common factor of any reasonable definition is that,
in the absence of duties, all the participants must have reached their goals
(see Th.~\ref{th:agreement}).
Here we focus on a definition of duties 
where $\vdash$ is prioritized over $\Vdash$,
i.e.\ an event may be performed on credit only if no other ways
are possible.
% Such duties will depend on the agreement $C$ which has been reached
% among the participants.
More precisely, an event $e$ belongs to $\duties{\pmv A}{X}$ if 
$(i)$ $e$ is not already present in $X$, but is in some configuration $C$, 
$(ii)$ $\princ{e} = {\pmv A}$, and 
$(iii)$ either $e$ is $\vdash$-enabled by $X$, or,
if no $\vdash$-enablings are possible from $X$,  
then $e$ is $\Vdash$-enabled by some events in $C \cup X$.

\begin{definition} \label{def:duties}
For all ${\pmv A}$, 
% $C \in \conf[\cname]$, and 
for all $X$, 
we define $\duties{\pmv A}{X}$ as the set of events 
$e \not\in X$ such that $\princ{e} = {\pmv A}$ and
there exists $C \in \conf[\cname]$ such that $e \in C$, and
either $X \vdash e$ or 
$\nexists e' \in C \setminus X : X \vdash e' \,\land\, \exists D \subseteq C \cup X : D \Vdash e$.
A participant {\pmv A} is \emph{culpable} in $X$ when 
{\pmv A} has some duties in $X$.
\end{definition}

\begin{example} \label{ex:toys:culpable}
Recall the contract $\cname$ of Ex.~\ref{ex:toys:contract}.
% and let $C = \setenum{\atom a, \atom b, \atom c}$.
By Def.~\ref{def:duties},
in state $\emptyset$ only participant {\pmv C} is culpable,
with $\duties{\pmv C}{\emptyset} = \setenum{\atom c}$;
in $\setenum{\atom c}$ only {\pmv B} is culpable,
with $\duties{\pmv B}{\setenum{\atom c}} = \setenum{\atom b}$;
finally, in $\setenum{\atom b,\atom c}$ only {\pmv A} is culpable,
with $\duties{\pmv A}{\setenum{\atom b,\atom c}} = \setenum{\atom a}$.
\end{example}

\begin{example}
Let $\cname$ be a contract with 
$\setenum{\atom a_0,\atom a_1} \Vdash \atom a_2$,
$\setenum{\atom a_0,\atom a_2} \Vdash \atom a_1$,
$\setenum{\atom a_1,\atom a_2} \vdash \atom a_3$,
and $\emptyset \vdash \atom a_0$,
where $\princ{{\atom a}_i} = {\pmv A}_i$ for $i \in [0,3]$.
We have that only $\pmv A_0$ is culpable in $\emptyset$;
only $\pmv A_1$ and $\pmv A_2$ are culpable in $\setenum{\atom a_0}$;
only $\pmv A_1$ is culpable in $\setenum{\atom a_0,\atom a_2}$;
only $\pmv A_2$ is culpable in $\setenum{\atom a_0,\atom a_1}$;
only $\pmv A_3$ is culpable in $\setenum{\atom a_0,\atom a_1,\atom a_2}$;
finally, no one is culpable in $C = \setenum{\atom a_0,\atom a_1,\atom a_2, \atom a_3} \in \conf[\cname]$. 
\end{example}

The following theorem establishes that it is safe to execute contracts
after they have been agreed upon.
More precisely, in each state $X$ of the contract execution, 
either all the participant goals have been fulfilled,
or some participant is culpable in $X$.
Note that, in consequence of Def.~\ref{def:duties},
a participant can always exculpate herself by performing some of her duties.
This is because, if $D = \duties{\pmv A}{X}$ is not empty,
participant {\pmv A} is always allowed to perform all the events in $D$,
eventually reaching a state where she is not culpable
(note also that in the maximal state $\Ev$ no one is culpable). 

\begin{theorem} \label{th:agreement}
If an agreement on $\cname$ exists,
then for all participants ${\pmv A} \in \Princ$, and for all 
$X \subseteq \Ev$,
either $\ok{A}{X}$, or some participant  is culpable in $X$.
\end{theorem}

\hidden{
\begin{example} \label{ex:toys:reach}
Let us consider the example \ref{ex:toys:contract}, and suppose that now Carl wants to play with his own toy car and will not lend it anymore. We have the following contract:
%  \Ev = \setenum{\atom{a},\atom{b},\atom{c}}
%  \qquad
%  \princsym = \setenum{\bind{\atom a}{\pmv A},\bind{\atom b}{\pmv B},\bind{\atom c}{\pmv C}}
%
\[
  \setenum{\atom b} \vdash {\atom a}
  \qquad 
  \setenum{\atom c} \vdash {\atom b}
  \hspace{40pt}
  \ok{\pmv A}{\setenum{\atom b}}
  \qquad  \ok{\pmv B}{\setenum{\atom c}}
\]
We have that:
\[
\begin{array}{c}
\conf[\emptyset]{\cname} = \setenum{\emptyset}
\hspace{40pt}
\conf[X]{\cname} = \setenum{X, \Ev}
\quad \text{ if } |X| > 1
\\[8pt]
\conf[\setenum{\atom a}]{\cname} = \setenum{\setenum{\atom a}}
\hspace{20pt}
\conf[\setenum{\atom b}]{\cname} = \setenum{\setenum{\atom b},\setenum{\atom a,\atom b}}
\hspace{20pt}
\conf[\setenum{\atom c}]{\cname} = \setenum{\setenum{\atom c},\setenum{\atom b,\atom c}, \Ev}
\end{array}
\]
We have that none of the events  $\atom a$,$\atom b$,$\atom c$ are  $\emptyset$-reachable, because they all are related to the \emph{unpromised} event $\atom c$:''Carl will lend his toy car''.
\end{example}
}
\subsection{A logical characterisation of agreements} \label{sect:pcl-contracts}

The problem of deciding if an agreement exists on some contract $\cname$
is reduced below to the problem of proving formulae 
in the contract logic \pcl\!\!~\cite{BZ10lics}.
A comprehensive presentation of \pcl is beyond the scope of this paper,
so we give here a brief overview, and we refer the reader 
to~\cite{BZ10lics,PCLtr} for more details.

\pcl extends intuitionistic propositional logic IPC
with a new connective, called \emph{contractual implication} and denoted by~$\coimp$.
Differently from IPC,
a contract $\sf b \coimp a$ implies $\sf a$ not only when $\sf b$ is true, 
like IPC implication, but also in the case that a ``compatible'' 
contract, e.g.\ $\sf a \coimp b$, holds.
Also, \pcl is equipped with an indexed lax modality
$\says$, similarly to the one in~\cite{Garg08modal}.

The Hilbert-style axiomatisation of \pcl extend that of IPC    
with the following axioms:
\begin{align*}
& \top \coimp \top 
&& \phi \imp (A \says \phi) \\
& (\phi \coimp \phi) \imp \phi
&& (A \says A \says \phi) \imp A \says \phi \\
& (\phi' \imp \phi) \imp (\phi \coimp \psi) \imp (\psi \imp \psi') \imp (\phi' \coimp \psi')
&& (\phi \imp \psi) \imp (A \says \phi) \imp (A \says \psi) 
\end{align*}

The Gentzen-style proof system of \pcl extends that of IPC with the following rules 
(we refer to~\cite{BZ10lics} for the standard IPC rules, and for 
the rules for the $\says$ modality).
\[
  \irule
  {\Gamma \;\vdash\; q}
  {\Gamma \;\vdash\; p \coimp q}
  \qquad
  \irule
  {\Gamma,\ p \coimp q,\ a \;\vdash\; p \quad
   \Gamma,\ p \coimp q,\ q \;\vdash\; b}
  {\Gamma,\ p \coimp q \;\vdash\; a \coimp b}
  \qquad
  \irule
  {\Gamma,\ p \coimp q,\ r \;\vdash\; p \quad
  \Gamma,\ p \coimp q,\ q \;\vdash\; r}
  {\Gamma,\ p \coimp q \;\vdash\; r}
\]
Notice the resemblance between the last rule and the rule ($\imp$L) of IPC:
the only difference is that here we allow the conclusion $r$ to
be used as hypothesis in the leftmost premise.
This feature allows $\coimp$ to resolve circular assume/guarantee rules,
e.g.\ to deduce $a$ and $b$ from the formula
$a \coimp b \;\land\; b \coimp a$.

The proof system of \pcl enjoys cut elimination and the subformula property.
The decidability of the entailment relation $\vdash_{\pcl}$
is a direct consequence of these facts (see~\cite{BZ10lics} for details).

In Def.~\ref{def:contracts-to-pcl}
we show a translation from contracts to \pcl formulae.
In particular, our mapping is a bijection into the fragment of \pcl
(called \pclminus\!) which comprises 
atoms, conjunctions, says, and non-nested (standard/contractual) implications.

\begin{definition} \label{def:contracts-to-pcl} 
The mapping $[ \cdot ]$ from contracts into \pclminus formulae is defined as follows:
\[
\begin{array}{l}
    [(D_i \circ \atom{a}_i)_i ] = \bigwedge_{i} [D_i \circ \atom{a}_i] 
    \\[5pt]
   \left[ \setcomp{{\atom d}_i}{i \in {\mathcal I}} \circ  \atom a \right]
   =
     \princ{\atom a} \says ( \bigwedge_{i \in \mathcal{I}} \; \princ{\atom d_i} \says \atom{d}_i )  [\circ] \; \atom a 
   \end{array}
   \hspace{20pt} \mathit{where }\, 
    [\circ] = \begin{cases} 
      \imp & \text{if $\circ = \;\vdash$} \\
      \coimp & \text{if $\circ = \;\Vdash$}
   \end{cases}
\]
\end{definition}

\newcommand{\thespclreach}{
For all contracts $\cname$, an events $e$ is reachable in $\cname$
iff
$[\cname] \vdash_{\pcl} \princ{e} \says e$.
}
\begin{theorem}\label{th:es-pcl}
\thespclreach
\end{theorem}
\begin{proof}({\it Sketch})
We extend the definition of configuration, 
by allowing events to be picked from a set $X$, in the absence of their premises.
We say that $C \subseteq \Ev$ is an \emph{$X$-configuration} of $\cname$ iff  $X \subseteq C$
and 
\[
  \exists e_0, \ldots, e_n \in C. \  
  \setenum{e_0, \ldots, e_n} = C
  \;\land\; 
  \forall i \leq n.\; 
  \big(
  e_i \in X
  \;\lor\;
  \setenum{e_0, \ldots, e_{i-1}} \vdash e_i  
  \;\lor\;  
  C \Vdash e_i 
  \big)
\]
This allows, given an $X$-configuration, to add/remove any event
and obtain an $Y$-configuration, possibly with $Y \neq X$.
We shall say that the events in $X$ have been taken ``on credit'',
to remark the fact that they may have been performed in the absence of 
a causal justification.
Notice that Def.~\ref{def:conf} is the special case of the above when $X = \emptyset$.
An event $e$ is $X$-reachable if it belongs to some \mbox{$X$-configuration}.
For all $X$, we define the set $\mathcal{R}(X)$ by the following 
inference rules:
\[
\begin{array}{c}
  \irule
  {D \vdash e \hspace{12px} D \subseteq {\mathcal{R}}(X)}
  {e \in \mathcal{R}(X)} 
%  \hspace{4pt} [\vdash_{\mathcal{R}}]
\hspace{20pt}
  \irule
  {D \Vdash e \hspace{12px} D \subseteq {\mathcal{R}}(X \cup \setenum{e})}
  {e \in \mathcal{R}(X)} 
%  \hspace{4pt} [\Vdash_{\mathcal{R}}]
\hspace{20pt}
  \irule
  {e \in X}
  {e \in \mathcal{R}(X)} 
%  \hspace{4pt} [\subseteq_{\mathcal{R}}]
\end{array}
\]
The set $\mathcal{R}(X)$ is used as a bridge in proving that
$e$ is $X$-reachable iff $[\cname],\, X \vdash_{\pcl} e$.
We prove first that $\mathcal{R}(X)$ contains exactly the $X$-reachable events,
and then we prove that $[\cname],\, X \vdash_{\pcl} e$ iff $e \in \mathcal{R}(X)$. 
The actual inductive statement is a bit stronger.
For all conjunction of atoms $\varphi$ and for all sets of conjunctions of atoms $\Phi$,
we denote with $\mkset{\varphi}$ and $\mkset{\Phi}$ the sets of atoms
occurring in $\varphi$ and in $\Phi$, respectively.
Then, we prove that for all $\varphi$ and for all $\Phi$:
\(
  \mkset{\varphi} \subseteq \mathcal{R}(\mkset{\Phi})
  \iff
  [\cname], \Phi \, \vdash_{\pcl} \varphi
\).
The $(\Leftarrow)$ direction is proved by induction on the depth of the derivation of 
$[\cname], \Phi \, \vdash_{\pcl} \varphi$.
For the $(\Rightarrow)$ direction, we let $e \in \mkset{\varphi}$, 
and then we proceed by induction on the depth of the derivation of $e \in \mathcal{R}(\mkset{\Phi})$.
\end{proof}

The following theorem 
reduces the problem of deciding agreements to
provability of \pcl formulae.
Concretely, one can use the decision procedure of \pclminus to
compute the set $C$ of reachable events. 
Then, an agreement exists iff 
each principal {\pmv A} has some goals contained in $C$.

\begin{theorem} \label{th:agreement-iff-reachable}
A contract $\cname$ admits an agreement iff:
\[
  \forall {\pmv A} \in \Princ.\;\;
  \exists G \subseteq \Ev.\;\;
  \big(
  \ok{\pmv A}{G} \;\land\;  
  \forall e \in G : [\cname] \vdash_{\pcl} \princ{e} \says e
  \big)
\]
\end{theorem}
\begin{proof}
($\Rightarrow$)
Let $C$ be an agreement on $\cname$, and let $\Princ = \setenum{\pmv A_i}_i$. 
By Def.~\ref{def:agreement}, $\ok{\pmv A_i}{C}$ for all $i$.
By definition of $\oksym$, 
there exist $G_i \subseteq C$ such that $\ok{\pmv A_i}{G_i}$.
Since $G_i \subseteq C \in \conf[\cname]$, then $G_i$ is reachable.
Therefore, by Theorem~\ref{th:es-pcl}, 
$[\cname] \vdash_{\pcl} \princ{e} \says e$, for all $e\in G_i$.

($\Leftarrow$)
% Let us now demonstrate the converse, by constructing the agreement.
Let $\Princ = \setenum{\pmv A_i}_i$, and let $\setenum{G_i}_i$ be such that
$\ok{\pmv A_i}{G_i}$ and $[\cname] \vdash_{\pcl} \princ{e} \says e$ 
for all $i$ and for all $e \in G_i$.
By Theorem~\ref{th:es-pcl}, each $G_i$ is reachable.
By Lemma \ref{lem:all-in-conf}, for all $i$ 
there exists $C_i \in \conf[\cname]$  such that $C_i \supseteq G_i$.
By Lemma \ref{lem:conf-plus-conf}, $C = \bigcup_i C_i \in \conf[\cname]$ 
is an agreement on $\cname$. 
\end{proof}

Finally, note that also $\duties{\pmv A}{X}$ can be computed by 
exploiting the correspondence with \pcl\!.
More precisely, we use $\vdash_{\pcl}$ to compute the set of 
all reachable events, so obtaining the maximal configuration 
(Lemma~\ref{lem:conf:reachable-maximal}), and then to compute $D \vdash e$ 
as prescribed by Def.~\ref{def:duties}.

\section{Related work} \label{sec:related-work}

Contracts have been investigated using a variety of models, e.g.\
c-semirings~\cite{Buscemi07transactional,Buscemi07ccpi,Ferrari06logic},
behavioural types~\cite{Bravetti07sc,Carpineti06basic,Castagna09contracts},
logics~\cite{Artikis09jlap,PrisacariuS07formal},
\emph{etc}.
All these models do not explicitly deal with the circularity issue,
which instead is the focus of this paper.

Circularity is dealt with at a logical (proof-theoretic) level 
in~\cite{BZ10lics};
the relation between reachability in our model and provability 
in the logic of~\cite{BZ10lics} is stated by Theorem~\ref{th:es-pcl}.
Compared to~\cite{BZ10lics}, our model features a finer notion of duties:
while~\cite{BZ10lics} focusses on reachable events, Def.~\ref{def:duties}
singles out which events must be performed in a given state,
by interpreting $D \vdash e$ as 
``I will do $e$ \emph{after} $D$ has been done''.

In~\cite{Hvitved11jlap} a trace-based model for contracts is defined. 
Similarly to ours, a way is devised for blaming misconducts,
also taking into account time contraints. 
However,~\cite{Hvitved11jlap} is not concerned in how to
reach agreements, so 
the modeling of mutual obligations (circularity) is neglected. 
It seems interesting to extend our model with temporal deadlines,
which would allow for a tighter notion of agreement,
and, more in general, with soft constraints, 
which could be used to model QoS requirements.

\newcommand{\response}{\,\bullet\!\!\!\to}

In~\cite{Hildebrandt10places} a generalization of prime event
structures is proposed where a \emph{response} relation 
(denoted with~$\response$) is used to characterize the accepting traces
as those where, for each $a \,\response b$, if $a$ is present in the
trace, then $b$ eventually occurs after $a$. 
The response relation bears some resemblance with our $\Vdash$ relation,
but there are some notable differences.  
First, having $a\Vdash b$ 
does not necessarily imply that a configuration containing
$a$ must contain also $b$ (another enabling could have been used), 
whereas $a \, \response b$ stipulates
that once one has $a$ in an accepting configuration, then also $b$
must be present.  
Indeed, an enabling $a \Vdash b$ can be neglected, 
whereas $a \, \response b$ must be used. 
Also, augmenting the number of $\Vdash$-enablings
increases the number of configurations,
while adding more response relations reduces
the number of accepting configurations of the event structure.
Finally,~\cite{Hildebrandt10places} deals with conflicts,
while we have left this issue for future investigation.

\section{Conclusions}

We have proposed a basic model for contracts,
building upon a new kind of event structures 
which allow to cope with circular assume/guarantee
constraints.
Our event structures feature two enabling relations
(the standard enabling $\vdash$ of~\cite{Winskel86},
and the circular enabling $\Vdash$), but 
they lack a construct to model non-determinism,
and they only consider finite sets of events.
Some preliminary work on a generalisation of our event structures
with conflicts and infinite sets of events is reported in~\cite{BCPZ12ictcs}.
Further extensions to the basic model proposed here seem plausible:
for instance, more general notions of goals, agreements and duties.
Also, a formalisation of the intuitive notion of 
``participant protected by a contract'', which we used to motivate 
the circular enabling relation, seems most desirable.

Our contract model features an effective procedure for 
deciding when an agreement exists, 
and then for deciding the duties of participants at each execution step.
These procedures are obtained by the means of an encoding of contracts into
Propositional Contract Logic.
In particular, our encoding reduces the problem of detecting
whether an event is reachable, to that of proving a formula in \pcl\!.
The correctness of our encoding is stated in Theorem~\ref{th:es-pcl}.
An extension of such result is presented in~\cite{BCPZ12ictcs},
where configurations are characterised as provability of certain formulae in \pcl\!.

% We depart from the common principle that contracts are always respected
% after they are stipulated:  
% in our model, promises might not be always maintained,
% and when they are not, culpable participants can be identified.

A concrete usage scenario of our contract model 
is a protocol for exchanging, agreeing upon, and executing contracts.
In the initial phase of the protocol, 
a special participant {\pmv T} acts as a contract broker,
which collects the contracts from all the participants.
Then, {\pmv T} looks for possible agreements on subsets of the contracts
at hand.
After an agreement on $\cname$ has been found, {\pmv T}
shares a session with the participants in $\cname$.
As long as the goals of some participant have not been
fulfilled, {\pmv T} notifies the duties to each culpable participant.
Variants of this protocol are possible which dispose {\pmv T} from
some of his tasks.
Notice that reaching an agreement is an essential requirement
for the security of this protocol:
if an untrusted contract broker claims to have found an agreement when
there is none, then Theorem~\ref{th:agreement} no longer applies,
and a situation is possible where a participant has not reached her goals,
but no one is culpable.
Notably, participants can still protect themselved against untrusted brokers,
by always requiring in their contracts the suitable ($\vdash$ / $\Vdash$) preconditions.
This protocol can be formally described in  
the process calculus \coco\!~\cite{BTZ11ice}. 
This requires to specialise
the abstract contract model of \coco to the contracts presented
in this paper, and, accordingly, to make the observables in
$\fuse{}{}$/$\ask{}{}$ prefixes correspond to agreements/duties, respectively.
Static analyses on \coco, e.g.\ the one in~\cite{BTZ12coordination},
may then be used to detect whether a participant always respects the contracts she advertises.

\paragraph{Acknowledgments.}
This work has been partially supported by 
by Aut.\ Region of Sardinia under grants L.R.7/2007 CRP2-120 (Project TESLA) 
and CRP-17285 (Project TRICS).

\bibliographystyle{eptcs}
\bibliography{main}

\end{document}